\documentclass{article}

 \usepackage{amsmath}
\usepackage{amsfonts}
\usepackage{amssymb}
\usepackage{amsthm}
\usepackage[sort&compress,square,comma,numbers]{natbib}
 \usepackage[english]{babel}

\usepackage{lineno,hyperref}
\modulolinenumbers[5]

\newtheorem{theorem}{Theorem}
\newtheorem{lemma}{Lemma}
\newtheorem{corollary}{Corollary}

\begin{document}
\begin{center}
\bigskip

\textbf{The Allen--Uzawa elasticity of substitution for nonhomogeneous production functions}

\bigskip

Elena Burmistrova, Sergey Lobanov

Department of Mathematics, Faculty of Economic Sciences,
National Research University Higher School of Economics, Moscow, Russia

\bigskip
\end{center}

\textbf{Abstract. }This note proves that the representation of the Allen elasticity of substitution  obtained by Uzawa for linear homogeneous functions holds true for nonhomogeneous functions. It is shown that the criticism of the  Allen-Uzawa elasticity of substitution in the works of Blackorby, Primont, Russell is based on an incorrect example.

\bigskip



\section{Introduction}
In the book by \citet{Allen} the concept of the two-variable elasticity of substitution (ES)
proposed by \citet{Hicks}  has been extended to an arbitrary number of factors. Later, (the books by Allen and Hicks were first published in 1938 and 1932 respectively) \citet{Uzawa} noticed that for homogeneous functions of degree $k=1$, known as linear homogeneous functions, the value of Allen elasticity of substitution (AES) can be expressed in terms of the cost function and its derivatives. The vector of the input factors  $x\in \mathbb{R}^{n}_{++}$ from Allen's definition
needs to be linked to the price vector  $p\in \mathbb{R}^{n}_{++}$ from Uzawa's form of AES in such a way that the minimum cost of input factors for the given level of output $y>0$ is attained at the given point $x$.

In his paper Uzawa solves the problem of describing the class of all linear homogeneous functions with constant AES. Representation of AES for such functions through the cost function is an auxiliary result of this article, achieved by means of certain relationships established by Allen for linear homogeneous functions. Specifically, on page 292 of his paper Uzawa references page 598 of Allen's book containing the equality $F_0 = 0$, which is true for linear homogeneous functions, but does not hold even for homogeneous functions of degree $k=2$.

We  are not aware of any publication where Uzawa's form for the AES is provided with proof for nonhomogeneous functions. The main purpose of this article is to present such proof.

Section~\ref{Prelim} of this paper contains key definitions, as well as the concepts from mathematical analysis used in later  sections. Some lemmas and proofs from linear algebra have been relegated to appendices.

In section~\ref{AESinUzawa}, we formulate and prove the main result of the paper on Uzawa's form of AES.

In section~\ref{Morish}, it is shown that in examples from  \citet{Blackorby1989882, Blackorby2007203} the  production function for all points of $\mathbb{R}^{3}_{++}$
either has a singular bordered Hessian or has no partial derivatives. Thus, for such a function AES is not defined at all. Interestingly, the cost function is nevertheless defined for all $p \in \mathbb{R}^{3}_{++}$ and belongs to the class $C^{\infty} $.

The final section~\ref{Discus} references some publications that discuss  variations of ES which differ from AES.

\section{Preliminaries}\label{Prelim}
\subsection{The definition of the Allen elasticity of substitution (AES)}

If  $f\colon X \subset \mathbb{R}^n \to \mathbf{R}$ is a function of class $C^2$ in a neighborhood of some point $x$, then we denote by $f_i(x)$ and $f_{ij}(x)$  the first derivatives with respect to $x_i$, $i = 1,\dots,n$, and second order derivatives for the pairs of variables $x_i,x_j$, $i,j =1,\dots,n$, at that point. As usual, the gradient of $f$ at the point $x$ is a vector $\nabla f(x)$ of all first-order partial derivatives $f_i(x)$, the Hessian matrix $H_f(x)$ of $f$ at the point $x$ consists of elements $f_{ij}(x)$. The bordered Hessian $f$ at the point $x$ is the matrix of size $(n + 1)\times (n + 1)$ of the following form
\begin{equation}\label{bordH}
  \left(
  \begin{array}{ccccc}
    0 & f_1 & f_2 & \dots & f_n \\
    f_1 & f_{11} & f_{12} & \dots & f_{1n} \\
    \multicolumn{5}{c}{\dotfill} \\
    f_n & f_{n1} & f_{n2} & \dots & f_{nn}
      \end{array}
  \right)
  =
   \left(
  \begin{array}{c|c}
    0 & \nabla f(x) \\ \hline
    [\nabla f(x)]^T \vphantom{{T^T}^T}& H_f(x) \\
  \end{array}
  \right)
  .
\end{equation}
The determinant of (\ref{bordH}) is denoted by $F$; the cofactor of any element $f_{ij}$ in this matrix will be denoted by $F_ {ij}$. Allen assumes the positivity of all the first partial derivatives and the alternation of the signs of the leading principal minors
$\Delta_2, \dots, \Delta_{n + 1}$ of the bordered Hessian matrix, which is equivalent to requiring the strong concavity of $f$ (see. \Citet[p. 414] {Diewert1981397}). For our purposes it is sufficient to have a single requirement that $\Delta_{n + 1} =F$ is nonzero.

Following \citet[p. 504]{Allen}, the partial elasticity of substitution of factors $i$ and $j$ (or simply AES of factors $i$ and $j$) is the number
\begin{equation}\label{defAES}
\sigma^A_{ij}(x)=\frac{x_1f_1+\dots+x_nf_n}{x_ix_j}\frac{F_{ij}}{F}.
\end{equation}
 The bordered Hessian matrix is symmetric for functions of class $C^2$. Therefore, $\sigma^ A_{ij} =\sigma^ A_{ji}$ for all $i,j = 1,\dots, n$.

\subsection{The elasticity of substitution by Hicks and AES in the case of two variables}

In the case of two variables, the definition of (\ref{defAES}) is consistent with the definition of the elasticity of substitution  introduced by Hicks (HES). Indeed, in this case $\sigma^A_{12}(x_1,x_2)$ equals
\begin{equation}\label{defHES}
  \frac{x_1f_1+x_2f_2}{x_1x_2}\frac{-\left|
                                                \begin{array}{cc}
                                                  0 & f_1 \\
                                                  f_2 & f_{21} \\
                                                \end{array}
                                              \right|
  }{\left|
      \begin{array}{ccc}
        0 & f_1 & f_2 \\
        f_1 & f_{11} & f_{12} \\
        f_2 & f_{21} & f_{22} \\
      \end{array}
    \right|
  }
  = \frac{x_1f_1+x_2f_2}{x_1x_2}
  \frac{f_1f_2}{-f_{11}f_2^2+2f_{12}f_1f_2-f_{22}f_1^2},
\end{equation}
that matches the formula for the HES in \citet [p. 342]{Allen} up to notation.

In \citet{Blackorby1981147, Blackorby1989882, Blackorby2007203} one of the advantages of the Morishima elasticity of substitution (MES) is considered to be the possibility to represent MES as the so-called logarithmic derivative.
\begin{equation}\label{defMES}
  \sigma^M_{ij}=\frac{\partial\ln(C_i/C_j)}{\partial\ln(p_j/p_i)}.
\end{equation}
Preceding the discussion of this definition in section~\ref{Morish}, let us point to the possibility of such representation of $\sigma^A$ in the case of two factors.

At the level curve $f (x_1, x_2) = y$, where $y$ is a predetermined number, provided $f_1\neq 0$ it is possible to express variable $x_1$ as the function of class $C^2$ of variable $x_2$ per the implicit function theorem. While the derivative
$\frac{\partial x_1}{\partial x_2}= - \frac {f_2}{f_1}$
can be treated as the rate of substitution
of variable $x_1$ for variable $x_2$
consider the logarithmic derivative, which Hicks calls the elasticity of substitution
\begin{equation}\label{defHESln}
  \sigma^H(x_1,x_2)=\frac{\partial\ln(x_1/x_2)}{\partial\ln(f_2/f_1)}.
\end{equation}
The formula (\ref{defHES}) can be obtained from the formula (\ref {defHESln}) using the implicit function theorem for the following system of equations
\begin{equation*}
  \left\{
     \begin{array}{l}
       \ln(x_1/x_2)-u=0, \\
       \ln(f_2/f_1)-v=0,\\
       f(x_1,x_2)-y=0.
     \end{array}
   \right.
\end{equation*}
In vector notation, this is the equation of the form $\Phi(u, v, x_1, x_2)=0$. Columns of the Jacobi matrix of function $\Phi$
\begin{equation*}
  \left(
     \begin{array}{cccc}
       -1&0&\frac{1}{x_1}&-\frac{1}{x_2} \\
       0&-1&\frac{f_{21}}{f_2}-\frac{f_{11}}{f_1}&\frac{f_{22}}{f_2}-\frac{f_{12}}{f_1}\\
       0&0&f_1&f_2
     \end{array}
   \right),
\end{equation*}
corresponding to the partial derivatives with respect to $u, x_1, x_2$,
form the matrix of non-zero minor of order $3$. It is immediately verified that the
\begin{equation*}
  \left|
     \begin{array}{ccc}
       -1&\frac{1}{x_1}&-\frac{1}{x_2} \\
       0&\frac{f_{21}}{f_2}-\frac{f_{11}}{f_1}&\frac{f_{22}}{f_2}-\frac{f_{12}}{f_1}\\
       0&f_1&f_2
     \end{array}
   \right|
=-\frac{F}{f_1f_2},
\end{equation*}
where $F$ is the determinant of the bordered Hessian.

Here and later in this paper, we need to apply the implicit function theorem%
\footnote{See. \Citet[p. 490]{Zorich1}} for equations of the form $f(x,y)=0$, where  $f\colon M\subset \mathbf{R}^n\times \mathbf{R}^m\to \mathbf{R}^m$, $f\in C^k, k\ge 1$.
In such cases it is convenient to represent Jacobi matrix of $f$ as composed of two blocks $D_x f$ and $D_y f$, where $D_x f$ consists of the first $n$ columns of the Jacobi matrix and $D_y f$ consists of the last $m$ columns.
\begin{theorem}[implicit function theorem]
If a function $f:M\subset \mathbb{R}^n\times\mathbb{R}^m\to \mathbb{R}^m$ is of class  $C^k(M)$,
$k\ge1$,
$f(x_0,y_0)=0$,
and $D_y f(x_0, y_0)$ is an invertible matrix,
then there exist open sets
$V\subset\mathbb{R}^n$,
$W\subset\mathbb{R}^m$
and the mapping $\varphi:V\to W$ of class $C^k$
such that
$x_0\in V$, $y_0\in W$
and for all  $x\in V$, $y\in W$  the condition $f(x, y)=0$ is
equivalent to $y=\varphi(x)$,
i.e $f(x,\varphi(x))=0$ for all $x\in V$, and  $y=\varphi(x)$  is a unique solution of
the equation $f(x,y)=0$ with respect to $y$ for a given $x$. Moreover,
$$
D_x \varphi(x)=-[D_y f(x,y)]^{-1}D_x f(x,y) \text{  provided  } y=\varphi(x).
$$
\end{theorem}

Therefore, in the equation $\Phi(u,v,x_1,x_2)=0$ the vector of variables $(u,x_1,x_2)$ can be expressed through the variable $v$ and
\begin{equation*}
\left(
  \begin{array}{c}
    D_v u \\
    D_v x_1 \\
    D_v x_2 \\
  \end{array}
\right)
=-
  \left[
     \begin{array}{ccc}
       -1&\frac{1}{x_1}&-\frac{1}{x_2} \\
       0&\frac{f_{21}}{f_2}-\frac{f_{11}}{f_1}&\frac{f_{22}}{f_2}-\frac{f_{12}}{f_1}\\
       0&f_1&f_2
     \end{array}
   \right]^{-1}
\left(
  \begin{array}{r}
    0\\
    -1\\
    0\\
  \end{array}
\right)
.
\end{equation*}
Hence, the expression for $D_v u$ obtained here matches (\ref{defHES}).

\subsection{Homogeneous functions}
The function $f\colon X\subset \mathbf{R}^n\to \mathbf{R}$ is called homogeneous of degree $k\in \mathbf{R}$
if for all numbers  $t>0$ and any $x\in X$   $tx\in X$ and $f(tx)=t^k f(x)$. If the equality $f(tx)=t^k f(x)$ holds for all numbers $t$ in some neighborhood of $1$ in some neighborhood of $x$, then the function is called locally homogeneous of degree $k\in \mathbf{R}$ function in the neighborhood.

For homogeneous and locally homogeneous degree $k$ functions of class $C^1$ the Euler identity is true
$$
x_1f_1(x)+\dots+x_nf_n(x)=kf(x).
$$
The partial derivatives of these functions are homogeneous or locally homogeneous functions of degree $(k-1)$.

Homogeneous degree $k = 1$ functions are often called linear homogeneous functions. Euler's identity for the first derivatives of linear homogeneous functions of two variables  $x_1f_{11}+x_2f_{12}=0$, $x_1f_{21}+x_2f_{22}=0$ allows to convert the formula (\ref{defHES}) to the form
\begin{equation}\label{homHES}
 \sigma^H(x_1,x_2)=\frac{f_1(x_1,x_2) f_2(x_1,x_2)}{f(x_1,x_2) f_{12}(x_1,x_2)}.
\end{equation}
\citet[p. 343]{Allen} notes that it is in the form (\ref{homHES}) that the elasticity of substitution was first defined by
\citet[pp. 117,245]{Hicks}.

\subsection{The critical points of the Lagrange function in the problem of minimizing the cost of a given level of output}
The function
\begin{equation}\label{LagC}
  L(x,\lambda)=\sum\limits_{i=1}^n p_ix_i+\lambda(f(x)-y)
\end{equation}
is the Lagrange function of the problem of minimizing $\sum\limits_{i=1}^n p_ix_i$ subject to $f (x)-y=0$ where the number $y$ and vector $p\in \mathbf{R}^n_{++}$ are given parameters. Critical points of the Lagrange function (\ref{LagC}) are defined as the solutions $(x, \lambda)$ of equations
$$
\left\{
  \begin{array}{c}
    D_x L=0, \\
    D_\lambda L=0 \\
  \end{array}
\right.
\quad
\text{or}
\quad
\left\{
  \begin{array}{l}
    p_i+\lambda f_i(x)=0,\; i=1,\dots,n, \\
    f(x)-y=0. \\
  \end{array}
\right.
$$

For any point $\bar{x}\in \mathbf{R}^n_{++}$ at $p_i=f_i(\bar{x})$ and $y=f(\bar{x})$ the point$(\bar{x},\bar{\lambda})$, where $\bar{\lambda}=-1$, is a critical point of the Lagrange function (\ref {LagC}).

\begin{theorem}[on the dependence of the critical points on the parameters]\label{zavis}
Let $(\bar{x},\bar{\lambda})\in \mathbf{R}^n_{++}\times \mathbf{R}$ be a critical point of the Lagrange function \textup{(\ref{LagC})} for some vector $p=\bar{p}\in \mathbf{R}^n_{++}$ and the number $y=\bar{y}$, the function $f\colon X\subset \mathbf{R}^n\to \mathbf{R}$ is in the class $C^2$ in a neighborhood of $\bar{x}$ and the determinant of the bordered Hessian at this point is not zero.
Then there exist a neighborhood $V$ of the point $(\bar{p},\bar{y})$, a neighborhood $W$ of the point  $(\bar{\lambda},\bar{x})\in  \mathbf{R}\times \mathbf{R}^n_{++}$ and a bijection $\varphi\colon V\to W$ of class $C^1$ such that for all $(p,y)\in V$ the point
$(\lambda,x)=\varphi(p,y)$ is the only critical point of $L$ in $W$. Moreover, the Jacobi matrix of the mapping $\varphi$ satisfies the condition
\begin{equation}\label{diffKrit}
\left(
  \begin{array}{c|c}
    D_p \lambda & D_y \lambda\\  \hline
   D_p x& D_y x \\
  \end{array}
  \right)
=
-
\left(
  \begin{array}{c|c}
   0 & \nabla f\\  \hline
   [\nabla f]^T & \lambda H_f\\
  \end{array}
  \right)^{-1}
\left(
  \begin{array}{c|c}
   0 & -1\\  \hline
   E & 0\\
  \end{array}
  \right).
\end{equation}
\end{theorem}
\begin{corollary}\label{corxpop}
Under the conditions of the theorem~\textup{(\ref{zavis}\textup)} for all $i,j=1,\dots,n$
\begin{equation}\label{xpop}
  \frac{\partial x_i}{\partial p_j}=-\frac{F_{ij}}{\lambda F}.
\end{equation}
\end{corollary}
For the proof of Theorem~(\ref {zavis}) and Corollary~(\ref{corxpop}) see \ref{AppII}.

\section{Uzawa's form of AES}\label{AESinUzawa}

\begin{theorem}[on the Uzawa form for the AES]
Assume the function  $f$ satisfies the conditions of theorem~\textup{(\ref{zavis})}, function
$
C(p,y)=\sum\limits_{i=1}^n p_i x_i(p,y)
$
is defined by relation $(\lambda,x)=\varphi(p,y)$ from the conclusion of the theorem~\textup{(\ref{zavis})}.
Then the function $C(p,y)$ is of  the class $C^2$, is linearly homogeneous in $p$ and
\begin{equation}\label{usavaES}
  \sigma^A_{ij}(x(p,y))=\frac{C(p,y)C_{ij}(p,y)}{C_i(p,y)C_j(p,y)}.
\end{equation}
\end{theorem}
\begin{proof}
Because of the uniqueness of the critical points, if $(p, y)\in V$ and $(\lambda,x)=\varphi(p,y)$, then $(t\lambda,x)=\varphi(tp, y)$ for all $t$ close to unity. Therefore, $x$ is a locally homogeneous function of $p$ of degree  $0$, and the function
\begin{equation}\label{CostF}
C(p,y)=\sum\limits_{i=1}^n p_i x_i (p,y)
\end{equation}
is a locally homogeneous function of the first degree in $p$.

By applying the envelope theorem to the function
\begin{equation}\label{envelC}
C(p,y)=L(x,\lambda,p,y)=\sum\limits_{i=1}^n p_ix_i+\lambda(f(x)-y),
\end{equation}
where $(\lambda,x)=\varphi(p,y)$, and differentiating  identities~(\ref{envelC}) with respect to the variables $p$ and $y$ we obtain
\begin{align*}
D_p C & =D_p L \quad
\text{ or  }
C_i(p,y)=x_i(p,y)
\text{ for all $i=1,\dots,n$} \\
D_y C & =DL_y=-\lambda(p,y).
\end{align*}
Therefore, the function $C$ belongs to the class $C^2$ on the set of variables and
$C_{ij}=\frac{\partial x_i}{\partial p_j}$  for all $i,j=1,\dots,n$.

Expressing $f_i$ of the equation  $p_i+\lambda f_i=0$, expression (\ref{defAES}) can be transformed into
\begin{multline*}
\sigma^A_{ij}=\frac{x_1\left(-\frac{p_1}{\lambda}\right)+\dots+x_n\left(-\frac{p_n}{\lambda}\right)}{x_ix_j}\frac{F_{ij}}{F}=\\
=\frac{C}{x_ix_j}\left(-\frac{F_{ij}}{\lambda F}\right)
=\frac{C}{x_ix_j}\frac{\partial x_j}{\partial p_i}=\frac{C C_{ji}}{C_iC_j}=\frac{C C_{ij}}{C_iC_j}.
\end{multline*}
\end{proof}

\begin{corollary}[on the relation between HES for production function and HES for the cost function]
Under the conditions of Theorem~\textup{(\ref{zavis})} for $n = 2$ HES for production function $f$ and HES for the cost function $C$ provided $C_{12}\neq 0$ are related by
$$
\sigma^H_f=\frac{1}{\sigma^H_C}.
$$
\end{corollary}
\begin{proof}

Because of the homogeneity of $C$ on $(p_1, p_2)$ by the formula (\ref{homHES})
$$
\sigma^H_C=\frac{C_1 C_2}{C C_{12}}.
$$
Applying formula (\ref{usavaES}) to the above result proves Corollary 2.
\end{proof}

\section{The elasticity of substitution by Morishima and the criticism of AES}\label{Morish}
In \citet{Blackorby1981147,Blackorby1989882,Blackorby2007203}
Morishima elasticity of substitution%
\footnote {The article \citet{Morishima1967144} was never translated from Japanese, Blackorbi and Russell argue that they independently came to the MES in 1975}
(MES) is defined, similarly to Uzawa form of AES, using the cost function $C(p, y)$ according to formula (\ref{defMES}). By the envelope theorem $C_i=x_i$, $C_j=x_j$, and $p_j/p_i=f_j/f_i$ at critical points of the problem of minimizing the cost of a given level of output $y$ at factors' prices $p$. Therefore, it may seem that the logarithmic derivative (\ref{defHESln}) is calculated anew from the definition of HES.  However, this is not true in the case of three or more input factors, as $C(p, y)$ is assumed to be the extreme value of the problem with the full range of factors $x$ and  full range of prices $p$. The use of such a function $C(p,y)$ makes keeping relation $f(x) = y$ in calculating the $\sigma^M_{ij}$ unnecessary.

The logarithmic derivative (\ref{defMES}) allows for different interpretations, leading to different values of the derivative.

We apply the implicit function theorem to the system of equations
\begin{equation*}
  \left\{
     \begin{array}{l}
       \ln(C_i/C_j)-u=0, \\
       \ln(p_j/p_i)-v=0.\\
     \end{array}
   \right.
\end{equation*}
This is a vector equation of the form $\Phi(u,v,p,y)=0$, where Jacobi matrix of the mapping $\Phi$  contains the matrix
$
\left(
   \begin{array}{rr}
     -1 & 0 \\
     0 & -1 \\
   \end{array}
 \right)
$
in the first two columns,
while the remaining columns have zeros in the second row, except for the columns of the partial derivatives with respect to  the variables $p_i$ and $p_j$. To give meaning to differentiation of $u$ with respect to $v$, the matrix of non-zero minor in the implicit function theorem should consist of a pair of columns for the variables $u, p_i$  or the variables $u, p_j$.

In the first case the variables $u, p_i$ can be represented as functions of class $C^1$ on the other variables of this problem, and
\begin{equation}\label{diffMES1}
\left(
  \begin{array}{c}
    D_v u\\
   D_v p_i\\
  \end{array}
  \right)
=
-
\left(
  \begin{array}{c|c}
   -1 & \frac{C_{ii}}{C_i}-\frac{C_{ij}}{C_j}\\
   0& -\frac{1}{p_i}\\
  \end{array}
  \right)^{-1}
\left(
  \begin{array}{r}
   0 \\
   -1\\
  \end{array}
  \right).
\end{equation}
Hence
\begin{equation}\label{MESj}
 D_v u=p_i\left(\frac{C_{ij}}{C_j}-\frac{C_{ii}}{C_i}\right).
\end{equation}
In the second case,
\begin{equation}\label{MESi}
 D_v u=p_j\left(\frac{C_{ij}}{C_i}-\frac{C_{jj}}{C_j}\right).
\end{equation}
If there are only two inputs, using linear homogeneity of $C$ in $p$ and applying Euler's identity it is possible to transform (\ref{MESj}), (\ref{MESi}) into the Uzawa form of AES. For three or more factors the values of (\ref {MESj}) and (\ref {MESi}) are, generally speaking, different. In \citet{Blackorby1981147,Blackorby1989882,Blackorby2007203} Blackorbi, Primont and Russell believed $\sigma^M_{ij}$ to be equal to the value (\ref{MESj}). Thus, unlike AES, MES is not symmetric. MES also can not be considered as a true logarithmic derivative.
It is inaccurate to indicate the representation of MES as a logarithmic derivative as an advantage of MES over AES.

Here is an example of a true logarithmic derivative
\begin{equation*}
  \varepsilon_{ij}=\frac{\partial\ln x_i}{\partial\ln p_j}=\frac{\partial x_i}{\partial p_j}\frac{p_j}{x_i}=C_{ij}\frac{p_j}{C_i}.
\end{equation*}
Hence, MES can be represented as the difference
$\varepsilon_{ji}-\varepsilon_{ii}$
between two logarithmic derivatives with respect to two different variables.

The inference in  \citet[p. 882]{Blackorby1989882},  \citet[p. 203]{Blackorby2007203} that AES for three or more factors does not possess any of the essential properties of the original concept of HES is based on the examples of the production function
\begin{equation}\label{antiAESf}
f(x_1,x_2,x_3)=\min(x_1,x_2^{1/2}x_3^{1/2})
\end{equation}
in \citet{Blackorby1989882} and a similar function in \citet{Blackorby2007203}.
AES is calculated in Uzawa's form using the cost function $C(p_1,p_2,p_3,y)=y\left(p_1+2p_2^{1/2}p_3^{1/2}\right)$ corresponding to the function (\ref{antiAESf}).

The result is an example of a function $f$, for which Uzawa's form of AES is defined, but the original AES is not. The proposed function has a singular bordered Hessian in the majority of the points $x\in \mathbf{R}^3_{++}$, and does not have even the first partial derivatives in points with  $x_1^2=x_2x_3$. The minimum of $p_1x_1+p_2x_2+p_3x_3$ on the surface $f(x_1,x_2,x_3)=y$ is achieved at the points of the curve, defined by the system of equations
\begin{equation}\label{kriv}
  \left\{
     \begin{array}{l}
       x_1-y=0, \\
       x_2x_3-y^2=0.\\
     \end{array}
   \right.
\end{equation}
Bordered Hessian for this problem is nonsingular
$$
\left|
  \begin{array}{cc|ccc}
    0 & 0 & 1 & 0 & 0 \\
    0 & 0 & 0 & x_3 & x_2 \\\hline
    1 & 0 & 0 & 0 & 0 \\
    0 & x_3 & 0 & 0 & 1 \\
    0 & x_2 & 0 & 1 & 0 \\
  \end{array}
\right|=-2x_2x_3.
$$
Therefore, the minimum value function is subject to the envelope theorem. It is immediately verified that $C_1^2=C_2C_3$ for all $(p,y)$, i.e. the minimum is attained on the curve (\ref{kriv}).

Criticism of a deeply rooted belief%
\footnote{For example, Varian wrote on the page 13 of the textbook \citet{Varian} ``The elasticity of substitution measures the curvature of an isoquant'', and further ``This is a relatively natural measure of curvature''. It is stated in the book \citet{Intriligator} on page 182 that ``The elasticities of substitution characterize the curvature of the isoquant''.}
that AES is a measure of the curvature of the isoquant $f(x)=y$ can be found in \citet{DeLaGrandville199723}. It is demonstrated using simple examples of the two-variable Cobb-Douglas functions which, is well known to have the AES constant and equal to one along the isoquant, while the curvature can vary widely. On the other hand, any curve with a constant curvature (i.e., an arc of a circle) is not a level curve of a CES-function. This indicates that a parallel shift of the isoquant along one of the axes does not change the curvature, but changes the AES. Furthermore, AES is not dependent on the units for each factor, whereas the curvature is.

Since AES and MES match HES in a case of two variables, the aforementioned statement about the curvature of the isoquant is valid for MES.
In \citet{Blackorby1989882}  $\sigma^M_{12}=0$  and $\sigma^M_{21}=1/2$ are calculated for function (\ref{antiAESf}), but it is not stated how this relates to any curvature. Here the surface $f(x_1,x_2,x_3)=y$ is a union of the sets $\{(x_1,x_2,x_3)\in \mathbf{R}^3_{++}\colon x_1\le y,x_2x_3=y^2\}$ and $\{(x_1,x_2,x_3)\in \mathbf{R}^3_{++}\colon x_1= y,x_2x_3\ge y^2\}$ with common points on the curve (\ref{kriv}).

\section{Discussion}\label{Discus}
The concept of elasticity of substitution was introduced into economic theory in the early 1930s  (see \citet{Molina}).
Constant elasticity of substitution (CES) functions were developed and studied by renowned economists
\citet{Hicks}, \citet{Arrow}, \citet{McFadden63}.

The research of possibilities to extend  HES to the case of many factors at the expense of ``freezing'' the value of the output, all but two of the prices and production factors, resulted in the introduction of concepts of DES (Direct partial elasticity of substitution) and SES (Shadow partial elasticity of substitution) in \citet{McFadden63}. The description of the classes of functions with a constant DES and a constant  SES derived by McFadden revealed the need to find other measures of substitution, accomodating a wider class of production functions with constant ES.

One of the reasons for AES modification is associated with the transition from the problem of minimum cost to the problem of maximum profit. If $p$, as before, is the vector of input prices and $p_y$ is the output price then per
\citet{Bertoletti2005183},
the Hotelling-Lau elasticity (HLES) is defined as
\begin{equation*}
  \sigma^{HL}_{ij}(p,p_y)=-\frac{\pi\pi_{ij}}{\pi_i\pi_j}, \text{ where }
\pi(p,p_y)=\max_{x\in \mathbf{R}^n_{++}}\left\{p_y f(x)-\sum_{i=1}^n p_ix_i\right\}.
\end{equation*}
This measure does not require the constancy of production or any other parameter that depends on the number of factors or prices. Such ES is called gross ES (the opposite being net ES). Examples of a net ES are AES and MES; gross ES version of MES (MGES) is introduced in \citet{Davis1996173}
\begin{equation*}
  \sigma^{MG}_{ij}(p,p_y)=p_i\left(\frac{\pi_{ij}}{\pi_j}-\frac{\pi_{ii}}{\pi_i}\right).
\end{equation*}
In \citet{Blackorby2007203} there is a proof that MES at the point $(p,y)$ is identical for all the pairs of factors to MGES at the point $(p,p_y)$, where $y=f(x(p,p_y))$ is the profit-maximizing output, if and only if the production function is homothetic, i.e. it is a composition of linear homogeneous function with an increasing external function. The authors again do not impose requirements on the smoothness of $f$ and nondegeneracy of $H_f$. In the example (\ref{antiAESf}) the values of  $\pi(p,p_y)$ are defined only at $p_y=p_1+2\sqrt{p_2p_3}$ and are equal to zero.

In their review papers \citet{Stern2004,Stern2011,Mundra} mention more than ten variations of ES or the elasticity of complementarity (EC), and describe the results of estimates of some of them using the four factors data set. It turns out, for example, that for some pairs of factors, AES and MES have opposite signs: complementary factors according to AES are substitutes according to MES.
It is appropriate to quote from \citet{Stern2004} `There are many different legitimate definitions of the ES and the elasticity of complementarity (EC). None of these is the one true ES --- which one is useful depends on what we wish to measure. As their value and even sign can vary dramatically, the choice of the appropriate indicator is important'.

\begin{center}
\section*{Appendices}
\end{center}

\appendix
\renewcommand{\thesection}{Appendix \Roman{section}}
\section{}\label{AppI}
\begin{lemma}\label{LbordA}
If $A$ is a square matrix of order $n$, $B$ is a $1\times n$ matrix, $C$ is a  $n\times 1$ matrix, then for all $\lambda\in \mathbf{R}$
\begin{equation}\label{bordA}
  \det
\left(
  \begin{array}{c|c}
    0 & B \\  \hline
   C& \lambda A \\
  \end{array}
  \right)
  =
\lambda^{n-1}  \det
\left(
  \begin{array}{c|c}
    0 & B \\  \hline
   C& A \\
  \end{array}
  \right)
.
\end{equation}
\end{lemma}
\begin{proof}
Assume $Q$ and $M$ are the matrices on the left side and on the right side of this equation respectively.
Let us enumerate the rows and the columns of these matrices with the numbers $0,1,\dots,n$.
Using the expansion of the determinant along the first row, we get
$$
\det(Q)=\sum\limits_{j=1}^n b_j(-1)^{1+j+1}\det(Q_{0j}).
$$
The expansion of  $\det(Q_{0j})$ along the left column gives
\begin{multline*}
\det(Q_{0j})=\sum\limits_{i=1}^n c_i(-1)^{i+1}\det(\lambda A_{ij})= \\
=\lambda^{n-1} \sum\limits_{i=1}^n c_i(-1)^{i+1}\det(A_{ij})=
\lambda^{n-1} \det(M_{0j}),
\text{ for all } j=1,\dots,n,
\end{multline*}
where $A_{ij}$ is the submatrix formed by deleting the $i$-th row and the $j$-th column of $A$. Hence $\det(Q)=\lambda^{n-1}\det(M)$.
\end{proof}

\begin{lemma}[on the multiplication of block matrices]\label{blpckM}
If the matrix $A$ consists of a left block $A^{(1)}$ of size $m\times n_1$ and the right block $A^{(2)}$ of size $m\times n_2$, and the matrix $B$ consists of an upper block $B^{(1)}$ of size $n_1\times p$ and the lower block $B^{(2)}$ of size $n_2\times p$, then $AB=A^{(1)}B^{(1)}+A^{(2)}B^{(2)}$.
\end{lemma}
\begin{proof}
By the definition of matrix multiplication
$$
(AB)_{ij}=\sum\limits_{k=1}^{n_1} a_{ik}b_{kj}+\sum\limits_{k=n_1+1}^{n_1+n_2} a_{ik}b_{kj}=(A^{(1)}B^{(1)})_{ij}+(A^{(2)}B^{(2)})_{ij}.
$$
\end{proof}

\section{}\label{AppII}

\begin{proof}[The proof of the theorem~\textup{(\ref{zavis}\textup)}]
We denote the left sides of equations
$$
\left\{
  \begin{array}{c}
    D_\lambda L=0, \\
    D_x L=0 \\
  \end{array}
\right.
$$
through $\Phi_1(\lambda,x,p,y)$ and $\Phi_2(\lambda,x,p,y)$, and then apply the implicit function theorem to the equation
$$
\Phi(\lambda,x,p,y)=\left(
                      \begin{array}{c}
                       \Phi_1(\lambda,x,p,y) \\
                       \Phi_2(\lambda,x,p,y) \\
                      \end{array}
                    \right)
=
\left(
  \begin{array}{c}
    0 \\
    0 \\
  \end{array}
\right).
$$
For convenience, let us represent the Jacobi matrix of the mapping $\Phi$ as a block matrix
$$
\Phi'(\lambda,x,p,y)=
\left(
  \begin{array}{c|c|c|c}
    D_\lambda \Phi_1 & D_x \Phi_1& D_p \Phi_1& D_y \Phi_1\\  \hline 
   D_\lambda \Phi_2 & D_x \Phi_2& D_p \Phi_2& D_y \Phi_2\\
  \end{array}
  \right),
$$
where $D_\lambda \Phi_1=0$, $D_x \Phi_1=\nabla f$, $D_p \Phi_1=0$, $D_y \Phi_1=-1$,
 $D_\lambda \Phi_2=[\nabla f]^T$, $D_x \Phi_2=\lambda H_f$, $D_p \Phi_2=E$, $D_y \Phi_2=0$.

Thus, the first $(n+1)$  columns of the matrix $\Phi'$ make up  the matrix
$$
Q=\left(
  \begin{array}{c|c}
   0 & \nabla f\\  \hline
   [\nabla f]^T & \lambda H_f\\
  \end{array}
  \right),
$$
whose determinant, according to Lemma~\ref{LbordA}, is equal to $\lambda^{n-1}F\neq 0$ due to the condition  $p_i+\lambda f_i=0$ for all  $i=1,\dots,n$ and positivity all $p_i$. Now the conclusion of the theorem follows directly from the implicit function theorem. Bijectivity of the mapping $\varphi$ follows from the reversibility of the Jacobi matrix (\ref{diffKrit}).
\end{proof}

\begin{proof}[Proof of the corollary~\textup{\ref{corxpop}}]
Multiplying matrices in (\ref{diffKrit}), we obtain
$$
-Q^{-1}\left(
  \begin{array}{c|r}
   0 & -1\\  \hline
   E & 0\\
  \end{array}
  \right)
=
\left(
  \begin{array}{c|c}
   -Q^{-1}\left(
           \begin{array}{c}
             0 \\
             E \\
           \end{array}
         \right)&
 -Q^{-1}\left(
           \begin{array}{r}
             -1 \\
             0 \\
           \end{array}
         \right)
  \end{array}
  \right)
$$
Applying Lemma~\ref{LbordA} to compute the cofactors of the elements of the matrix $Q$ with indices $i,j=1,\dots,n$,
we get
$$
\left(Q^{-1}\right)_{ij}=\frac{\det(Q_{ji})}{\det(Q)}=\frac{\lambda^{n-2}F_{ji}}{\lambda^{n-1}F}=\frac{F_{ij}}{\lambda F}.
$$
Hence, by Lemma~\ref{blpckM} the $(i,j)$th element of the matrix $D_p x$ is equal to  $-\frac{F_{ij}}{\lambda F}$ for all $i,j=1,\dots,n$.
\end{proof}


\begin{thebibliography}{20}
\providecommand{\natexlab}[1]{#1}
\providecommand{\url}[1]{\texttt{#1}}
\expandafter\ifx\csname urlstyle\endcsname\relax
  \providecommand{\doi}[1]{doi: #1}\else
  \providecommand{\doi}{doi: \begingroup \urlstyle{rm}\Url}\fi

\bibitem[{Allen, R. G. D.}(1962)]{Allen}
{Allen, R. G. D.}
\newblock \emph{{Mathematical Analysis for Economists}}.
\newblock {London: MacMillan and Co., Ltd.}, 2rd edition, 1962.

\bibitem[{Arrow K. J., Chenery H. B., Minhas B. S. and Solow R.
  M.}(1961)]{Arrow}
{Arrow K. J., Chenery H. B., Minhas B. S. and Solow R. M.}
\newblock {Capital-Labor Substitution and Economic Efficiency}.
\newblock \emph{{The Review of Economics and Statistics}}, 43\penalty0
  (3):\penalty0 225--250, 1961.

\bibitem[{Bertoletti, P.}(2005)]{Bertoletti2005183}
{Bertoletti, P.}
\newblock {Elasticities of substitution and complementarity: A synthesis}.
\newblock \emph{{Journal of Productivity Analysis}}, 24\penalty0 (2):\penalty0
  183--196, 2005.

\bibitem[Blackorby and Russell(1981)]{Blackorby1981147}
Blackorby and R.R. Russell.
\newblock The morishima elasticity of substitution: symmetry, constancy,
  separability, and its relationship to the hicks and allen elasticities.
\newblock \emph{Review of Economic Studies}, 48:\penalty0 147--158, 1981.

\bibitem[Blackorby et~al.(2007)Blackorby, Primont, and
  Russell]{Blackorby2007203}
C.~Blackorby, D.~Primont, and R.R. Russell.
\newblock The morishima gross elasticity of substitution.
\newblock \emph{Journal of Productivity Analysis}, 28\penalty0 (3):\penalty0
  203--208, 2007.

\bibitem[{Blackorby, C. and Russell R. R.}(1989)]{Blackorby1989882}
{Blackorby, C. and Russell R. R.}
\newblock {Will the Real Elasticity of Substitution Please Stand Up? (A
  comparison of the Allen/Uzawa and Morishima Elasticities)}.
\newblock \emph{{American Economic Review}}, 79\penalty0 (4):\penalty0
  882--888, 1989.

\bibitem[Davis and Shumway(1996)]{Davis1996173}
G.C. Davis and C.R. Shumway.
\newblock To tell the truth about interpreting the morishima elasticity of
  substitution.
\newblock \emph{Canadian Journal of Agricultural Economics}, 44\penalty0
  (2):\penalty0 173--182, 1996.

\bibitem[De~La~Grandville(1997)]{DeLaGrandville199723}
O.~De~La~Grandville.
\newblock Curvature and the elasticity of substitution: Straightening it out.
\newblock \emph{Journal of Economics/ Zeitschrift fur Nationalokonomie},
  66\penalty0 (1):\penalty0 23--34, 1997.

\bibitem[Diewert et~al.(1981)Diewert, Avriel, and Zang]{Diewert1981397}
W.E. Diewert, M.~Avriel, and I.~Zang.
\newblock Nine kinds of quasiconcavity and concavity.
\newblock \emph{Journal of Economic Theory}, 25\penalty0 (3):\penalty0
  397--420, 1981.

\bibitem[{Hicks, John R.}(1963)]{Hicks}
{Hicks, John R.}
\newblock \emph{{The Theory of Wages}}.
\newblock {St MartinТs Press, NY}, 2rd edition, 1963.

\bibitem[Intriligator(1971)]{Intriligator}
M.~D. Intriligator.
\newblock \emph{Mathematical Optimization and Economic Theory}.
\newblock Enclewood Cliffs: Prentice Hall, Inc., 1971.

\bibitem[{McFadden D.}(1963)]{McFadden63}
{McFadden D.}
\newblock {Constant Elasticity of Substitution Production Functions}.
\newblock \emph{{Review of Economic Studies}}, 30\penalty0 (2):\penalty0
  73--83, 1963.

\bibitem[{Molina, M.G.}(2005)]{Molina}
{Molina, M.G.}
\newblock {Capital theory and the origins of the elasticity of substitution
  (1932-35)}.
\newblock \emph{{Cambridge Journal of Economics}}, 29\penalty0 (3):\penalty0
  423--437, 2005.

\bibitem[Morishima(1967)]{Morishima1967144}
M.~Morishima.
\newblock A few suggestions on the theory of elasticity.
\newblock \emph{Keizai Hyoron (Economic Review)}, 16:\penalty0 144--150, 1967.

\bibitem[{Mundra, K.}(2013)]{Mundra}
{Mundra, K.}
\newblock {Direct and dual elasticities of substitution under non-homogenous
  technology and nonparametric distribution}.
\newblock \emph{{Indian Growth and Development Review}}, 6\penalty0
  (2):\penalty0 260--288, 2013.

\bibitem[Stern(2004)]{Stern2004}
D.I. Stern.
\newblock Elasticities of substitution and complementarity.
\newblock \emph{Rensselaer working papers in economics}, \penalty0
  (0403):\penalty0 53, 2004.

\bibitem[Stern(2011)]{Stern2011}
D.I. Stern.
\newblock Elasticities of substitution and complementarity.
\newblock \emph{Journal of Productivity Analysis}, 36\penalty0 (1):\penalty0
  79--89, 2011.

\bibitem[{Uzawa, H.}(1962)]{Uzawa}
{Uzawa, H.}
\newblock {Production functions with constant elasticities of substitution}.
\newblock \emph{{Review of Economic Studies}}, 29\penalty0 (4):\penalty0
  291--299, 1962.

\bibitem[{Varian, Hal R.}(1992)]{Varian}
{Varian, Hal R.}
\newblock \emph{{Microeconomic Analysis}}.
\newblock {New York: W. W. Norton \& Company}, 3rd edition, 1992.

\bibitem[{Zorich, Vladimir A.}(2004)]{Zorich1}
{Zorich, Vladimir A.}
\newblock \emph{{Mathematical Analysis {I}}}.
\newblock {Universitext}. {Springer-Verlag, Berlin}, 2004.
\newblock {Translated from the 2002 fourth Russian edition by Roger Cooke}.

\end{thebibliography}
\end{document}